\documentclass[journal,10pt]{IEEEtran}

\usepackage{verbatim,amssymb,amsmath,graphicx}
\usepackage{subfigure}
\IEEEoverridecommandlockouts

\newcommand{\FF}{\mathbb{F}}
\newcommand{\MM}{\mathbb{M}}

\newcommand{\ov}{\mathbf{o}}

\newcommand{\sv}{\mathbf{s}}
\newcommand{\xv}{\mathbf{x}}
\newcommand{\yv}{\mathbf{y}}

\newcommand{\hsrc}{{\rm HSRC}}
\newcommand{\nmax}{n_{max}}
\newcommand{\probup}{p_{node}}

\newcommand{\objup}{p_{obj}}
\newcommand{\divsrc}{\delta} % Number of distinct and mutually exclusive pairs of fragments which can be used to reconstruct a specific fragment. Depends on the diversity of the symmetric subset code.

\newtheorem{ex}{Example}
\newtheorem{lem}{Lemma}
\newtheorem{cor}{Corollary}
\newtheorem{defn}{Definition}

\begin{document}
\title{Homomorphic Self-repairing Codes for\\Agile Maintenance of Distributed Storage Systems}
\author{Fr\'ed\'erique Oggier and Anwitaman Datta
\thanks{F. Oggier is with Division of Mathematical Sciences,
School of Physical and Mathematical Sciences, Nanyang Technological University,
Singapore. A. Datta is with Division of Computer Science, School of Computer Engineering,
Nanyang Technological University, Singapore.
Email:\{frederique,anwitaman\}@ntu.edu.sg.
Part of this work appeared at Infocom 2011 \cite{OD11}.
}
}
\maketitle

%*************************************************************************%
%
% ABSTRACT
%
%************************************************************************%

\begin{abstract}
Distributed data storage systems are essential to deal with the need to store massive volumes of data. In order to make such a system fault-tolerant, some form of redundancy becomes crucial. There are various overheads that are incurred due to such redundancy - most prominent ones being overheads in terms of storage space and maintenance bandwidth requirements. Erasure codes provide a storage efficient alternative to replication based redundancy in storage systems. They however entail high communication overhead for maintenance in a networked setting, when some of the encoded fragments are lost due to failure of storage devices and need to be replenished in new ones. Such overheads arise from the fundamental need in storage systems to recreate (or keep separately) first a copy of the whole object before any individual encoded fragment can be generated and replenished. Traditional erasure codes, originally designed for communication over lossy channels, are optimized for recreation of the original message (object), but not for regeneration of individual lost encoded parts. We propose as an alternative a new family of erasure codes called \emph{self-repairing codes} (SRC) taking into account the peculiarities of distributed storage systems, specifically to improve the maintenance process. SRC has the following salient features: (a) encoded fragments can be repaired directly from other subsets of encoded fragments by downloading less data than the size of the complete object, ensuring that (b) a fragment is repaired from a fixed number of encoded fragments, the number depending only on how many encoded blocks are missing and independent of which specific blocks are missing. This paper lays the foundations by defining the novel self-repairing codes, elaborating why the defined characteristics are desirable for distributed storage systems. Then a concrete family of such code, namely, homomorphic self-repairing codes (HSRC) are proposed and various aspects and properties of the same are studied in detail and compared - quantitatively or qualitatively (as may be suitable) with respect to other codes including traditional erasure codes as well as other recent codes designed specifically for storage applications.

\end{abstract}
\begin{keywords}
 coding, networked storage, self-repair
\end{keywords}
%*************************************************************************%
%
% INTRODUCTION
%
%*************************************************************************%
\section{Introduction}

Various genres of networked storage systems, such as decentralized peer-to-peer storage systems, as well as dedicated infrastructure based data-centers and storage area networks, have gained prominence in recent years. Because of storage node failures, or user attrition in a peer-to-peer system, redundancy is essential in networked storage systems. This redundancy can be achieved using either replication, or (erasure) coding techniques, or a mix of the two. Erasure codes require an object to be split into $k$ parts, and mapped into $n$ encoded fragments, such that any $k$ encoded fragments are adequate to reconstruct the original object. Such coding techniques play a prominent role in providing storage efficient redundancy, and are particularly effective for storing large data objects and for archival and data back-up applications (for example, CleverSafe \cite{cleversafe}, Wuala \cite{wuala}).

Redundancy is lost over time because of various reasons such as node failures or attrition, and mechanisms to maintain redundancy are essential. It was observed in \cite{Liskov} that while erasure codes are efficient in terms of storage overhead, maintenance of lost redundancy entail relatively huge overheads. A naive approach to replace a single missing fragment will require that $k$ encoded fragments are first fetched in order to create the original object, from which the missing fragment is recreated and replenished. This essentially means that for every lost fragment, $k$-fold more network traffic is incurred.

Several engineering solutions can partly mitigate the high maintenance overheads. One approach is to use a `hybrid' strategy, where a full replica of the object is additionally maintained \cite{Liskov}. This ensures that the amount of network traffic equals the amount of lost data.\footnote{In this paper, we use the terms `fragment' and `block' interchangeably. Depending on the context, the term `data' is used to mean either fragment(s) or object(s).} A spate of recent works \cite{netcod,hierarchical} argue that the hybrid strategy adds storage inefficiency and system complexity, besides being a single point of bottleneck. Another possibility is to apply lazy maintenance \cite{TotalRecall,dattaP2P}, whereby maintenance is delayed in order to amortize the maintenance of several missing fragments. Lazy strategies additionally avoid maintenance due to temporary failures. Procrastinating repairs however may lead to a situation where the system becomes vulnerable, for example to bursty/correlated failures, and thus may require a much larger amount of redundancy to start with. Furthermore, the maintenance operations may lead to spikes in network resource usage \cite{dattaSSS}.

It is worth highlighting at this juncture that erasure codes had originally been designed in order to make communication robust, such that loss of some packets over a communication channel may be tolerated. Network storage has thus benefitted from the research done in coding over communication channels by using erasure codes as black boxes that provide efficient distribution and reconstruction of the stored objects. Networked storage however involves different challenges but also opportunities not addressed by classical erasure codes. Recently, there has thus been a renewed interest \cite{netcod,UCBsubm,hierarchical,biersackRGC,vijaykumarallerton,KLS,Shum-ICC} in designing codes that are optimized to deal with the vagaries of networked storage, particularly focusing on the maintenance issue. In a volatile network where nodes may fail, or come online and go offline frequently, new nodes must be provided with fragments of the stored data to compensate for the departure of nodes from the system, and replenish the level of redundancy (in order to tolerate further faults in future). In this paper, we propose a new family of codes called
\emph{self-repairing codes} (SRC), which are tailored to fit well typical networked storage environments.

As any linear $(n,k,d)$ erasure code over a $q$-ary alphabet, a SRC is formally a linear
map $c:\FF_{q^k}\rightarrow\FF_{q^n},~\sv\mapsto c(\sv)$
which maps a $k$-dimensional vector $\sv$ to an $n$-dimensional vector
$c(\sv)$. The set $C$ of codewords $c(\sv)$, $\sv\in\FF_{q^k}$,
forms the code (or codebook). The third parameter $d$ refers to the minimum
distance of the code: $d=\min_{\xv\neq\yv \in C} d(\xv,\yv)$ where the Hamming distance $d(\xv,\yv)$ counts the number of positions at which the coefficients of $\xv$ and $\yv$ differ. The minimum distance describes how many erasures can
be tolerated, which is known to be at most $n-k$, achieved by maximum distance
separable (MDS) codes. MDS codes thus allow to recover any codeword out of $k$
coefficients. Though SRCs are not MDS codes, their definition mimics the MDS property in terms of repair, namely,
we define the \emph{concept of self-repairing codes} as
$(n,k)$ codes designed to suit networked storage systems, that encode $k$
fragments of an object into $n$ encoded fragments to be stored at $n$ nodes,
with the properties that:\\
(a) \emph{encoded fragments can be
repaired directly from other subsets of encoded fragments by
downloading less data than the size of the complete object}. \\
More precisely, based on the analogy with the error correction capability of erasure codes, which is of any $n-k$ losses independently of which losses,\\ (b) \emph{a fragment can be repaired from a fixed number of encoded fragments, the number depending only on how many encoded blocks are missing and independently of which specific blocks are missing.}

Two families of SRC are known up to date \cite{OD11,OD2-11}.

To do so, SRCs naturally require more storage overhead than erasure codes for equivalent fault tolerance (static resilience). We will see more precisely later on that there is a tradeoff between the ability to self-repair
and this extra storage overhead: SRC could be tuned to be MDS at the price of losing the self-repair property, and conversely, the facility to self-repair can be adapted based on the amount of extra redundancy introduced. Consequently, SRCs can recreate the whole object
with $k$ fragments, though unlike for erasure codes, these are not arbitrary
$k$ fragments, however, many such $k$ combinations can be found (see Section
\ref{sec:static} for more details).

Note that even for traditional erasure codes, the property (a) may
coincidentally be satisfied, but in absence of a systematic mechanism this
serendipity cannot be leveraged. In that respect, hierarchical codes (HCs) \cite{hierarchical}
may be viewed as a way to do so, and are thus the closest example of
construction we have found in the literature, though they do not give any
guarantee on the number of blocks needed to repair given the number of losses,
i.e., property (b) is not satisfied, and has no deterministic guarantee for
achieving property (a) either. We may say that in spirit, SRC is closest to
hierarchical codes - at a very high level, SRC design features mitigate the
drawbacks of HCs.

While motivated by the same problem as regenerating codes (RGC) and HCs, that of efficient
maintenance of lost redundancy in coding based distributed storage systems,
the approach of self-repairing codes (SRC) tries to do so at a somewhat
different point of the design space. We try to minimize the number of nodes
necessary to reduce the reconstruction of a missing block, which automatically
translates into lower bandwidth consumption, but also lower computational
complexity of maintenance, as well as the possibility for
faster and parallel replenishment of lost redundancy. Thus SRCs allow \emph{light weight} (in terms of communication and computation overhead) and \emph{flexible} (in terms of flexibility in the number of options to carry out specific repairs, which in turn allow parallel and fast repairs), that is, \emph{agile} maintenance, of networked storage systems.

In this work, we make the following \emph{contributions}:\\
(i) We propose a new family of codes, self-repairing codes (SRC), designed specifically as an alternative to erasure codes (EC) for providing redundancy in networked storage systems, which allow repair of individual encoded blocks using only few other encoded blocks. Like ECs, SRCs also allow recovery of the whole object using $k$ encoded fragments, but unlike in ECs, these are not any arbitrary $k$ fragments. However, numerous specific suitable combinations exist.\\
(ii) We provide a deterministic code construction called \emph{Homomorphic Self-Repairing Code} (HSRC), showcasing that SRC codes can indeed be realized.\\
(iii) HSRC self-repair operations are computationally efficient. It is done by XORing encoded blocks, each of them containing information about all fragments of the object, though the encoding itself is done through polynomial evaluation (similar to popular ECs such as Reed-Solomon \cite{ReedSolomon} codes), not by XORing.\\
(iv) We show that for equivalent static resilience, marginally more storage is needed than traditional erasure codes to achieve self-repairing property.\\
(v) The need of few blocks to reconstruct a lost block naturally translates to low overall bandwidth consumption for repair operations. SRCs allow for both eager as well as lazy repair strategies for equivalent overall bandwidth consumption for a wide range of practical system parameter choices. They also outperform lazy repair with the use of traditional erasure codes for many practical parameter choices.\\
(vi) We show that by allowing parallel and independent repair of different encoded blocks, SRCs facilitate fast replenishment of lost redundancy, allowing a much quicker system recovery from a vulnerable state than is possible with traditional codes. This also implies a distribution of the repair related tasks across different nodes, thus avoiding bottlenecks or overloading any specific node.

%************************************************************************%
%
% RELATED WORK
%
%************************************************************************%
\section{Related work}

In \cite{netcod,UCBsubm}, Dimakis et al. propose regenerating codes (RGC) by
exposing the need of being able to reconstruct an erased encoded block from a
smaller amount of data than would be needed to first reconstruct the whole
object. They however do not address the problem of building new codes that
would solve the issue, but instead use classical erasure codes as a black box
over a network which implements random linear network coding and propose leveraging the properties of network
coding to improve the maintenance of the stored data. Network information flow based analysis shows the possibility to replace a missing fragment using network traffic equalling the volume of lost data. Unfortunately, it is possible to achieve this optimal limit only by communicating with all the $n-1$ remaining blocks. Consequently, to the best of our knowledge, regenerating codes literature generally does not discuss how it compares with engineering solutions like lazy repair, which amortizes the repair cost by initiating repairs only when several fragments are lost. Furthermore, for RGCs to work, even sub-optimally, it is essential to communicate with at least $k$ other nodes to reconstruct any missing fragment. Thus, while the volume of data-transfer for maintenance is lowered, RGCs are expected to have higher protocol overheads, implementation and computational complexity. For instance, it is noted in \cite{biersackRGC} that a randomized linear coding based realization of RGCs takes an order of magnitude more computation time than standard erasure codes for both encoding and decoding.
The work of \cite{vijaykumarallerton} improves on the original RGC papers in that instead of arguing the existence of regenerating codes via deterministic network coding algorithms, they provide explicit network code constructions.
Recently, collaborative RGC were introduced independently \cite{KLS,Shum-ICC}, where it was shown that collaboration among
new nodes joining the network and participating to the repair process can improve on traditional RGC, in terms of both
(i) storage-bandwidth trade-off, that is the amount of data that is stored at each node with respect to that which is downloaded by new nodes during repair, and (ii) number of simultaneous failures tolerated. While such analysis determines constraint on achievability, for classical as well as collaborative RGC, code constructions for collaborative RGC are even sparser, up to date, only one construction has been given in \cite{KLS}, which furthermore corresponds to the same repair
cost as erasure codes.

In \cite{hierarchical}, the authors make the simple observation that encoding two bits into three by XORing the two information bits has the property that any two encoded bits can be used to recover the third one. They then propose an iterative construction where, starting from small erasure codes, a bigger code, called hierarchical code (HC), is built by XORing subblocks made by erasure codes or combinations of them. Thus a subset of encoded blocks is typically enough to regenerate a missing one. However, the size of this subset can vary, from the minimal to the maximal number of encoded subblocks, determined by not only the number of lost blocks, but also the specific lost blocks. So given some lost encoded blocks, this strategy may need an arbitrary number of other encoded blocks to repair. Pyramid codes \cite{pyramid} explore similar ideas.

%*************************************************************************%
%
% HOMOMORPHIC CODES
%
%*************************************************************************%
\section{Homomorphic Self-Repairing Codes}

In what follows, we denote finite fields by $\FF$, and finite fields without the zero
element by $\FF^*$. The cardinality of $\FF$ is given by its index, that is, $\FF_2$ is
the binary field with two elements, which is nothing else than the two bits 0 and 1, with
addition and multiplication modulo 2, $\FF_q$ is the finite field with $q$ elements, and
$\FF_Q$ is the finite field with $Q$ elements.
If $q=2^t$, $Q=q^m=2^{tm}$, for some positive integers $m$ and $t$, an element $\xv\in\FF_Q$ can be represented by an $m$-dimensional vector $\xv=(\xv_1,\ldots,\xv_m)$ where $\xv_i\in\FF_q$, $i=1,\ldots,m$, by fixing a $\FF_q$-basis of $\FF_Q$.
Similarly, each coefficient $\xv_i$ can be written as $\xv_i=(x_{i1},\ldots,x_{it})$, $x_{ij}\in\FF_2$, so that $\xv$ may alternatively be seen as a $tm$-dimensional binary vector
$\xv=(x_{11},\ldots,x_{1t},\ldots,x_{m1},\ldots,x_{mt})$. We say that $\xv$ is a vector of size $m$ to refer to $m$ coefficients in $\FF_q$, so that $q$ determines the unit in which the size $m$ is measured: for example, if $q=2$, $\xv$ is $m$ bit long, if $q=8$, $\xv$ is $m$ bytes long.
To do explicit computations in the finite field $\FF_q$, it is often convenient to use the generator of the multiplicative group $\FF_{q}^*$, that we will denote by $w$.
A generator has the property that $w^{q-1}=1$, and there is no smaller
positive power of $w$ for which this is true.
Examples of finite fields that will be used later on are given in Table \ref{table:ff}.

\begin{table}
\[
\begin{array}{|l|l|ll|}
\hline
\FF_4    &  \FF_8        & \FF_{16} & \\
\hline
 0        &  0           & 0          & w^7 =  w^3 + w +1 \\
 1        &  1           & 1          &  w^8 =  w^2+1  \\
 w        &  w           & w          &  w^9 =  w^3+w \\
 w^2      &  w^2         & w^2        &  w^{10}=w^2+w+1 \\
 =w+1     &  w^3=w+1     & w^3        &   w^{11}=w^3+w^2+w\\
          &  w^4=w^2+w   & w^4=w+1    &  w^{12}=w^3+w^2+w+1\\
          &  w^5=1+w^2+w & w^5=w^2+w   &w^{13}=w^3+w^2+1\\
          &  w^6=w^2+1   & w^6=w^3+w^2 & w^{14}=w^3+1\\
\hline
\end{array}
\]
\caption{The finite fields $\FF_4,\FF_8$ and $\FF_{16}$, where $w$ denotes the generator of their respective
multiplicative groups $\FF_4^*$, $\FF_8^*$ and $\FF_{16}^*$.}
\label{table:ff}
\end{table}

%*************************************************************************%
\subsection{Encoding}

Let $\ov$ be an object of size $M$ to be stored over a network of $n$ nodes, that is $\ov\in\FF_{q^M}$, and let $k$ be a positive integer such that $k$ divides $M$. We can write
\[
\ov=(\ov_1,\ldots,\ov_k),~\ov_i\in\FF_{q^{M/k}}
\]
which requires the use of a $(n,k)$ code over $\FF_{q^{M/k}}$, that
maps $\ov$ to an $Mn/k$-dimensional vector $\xv$, or equivalently, an $n$-dimensional vector
\[
\xv=(\xv_1,\ldots,\xv_n),~\xv_i\in\FF_{q^{M/k}},
\]
after which each $\xv_i$ is given to a node to be stored. The theory developed below assumes this model, which can be adjusted to fit real scenarios as elaborated later in Example \ref{ex:real}.

Since the work of Reed and Solomon \cite{ReedSolomon}, it is known that linear
coding can be done via polynomial evaluation. In short, take an object
$\ov=(\ov_1,\ov_2,\ldots,\ov_k)$ of size $M$, with each $\ov_i$ in $\FF_{q^{M/k}}$,
and create the polynomial
\[
p(X)=\ov_1+\ov_2X+\ldots\ov_kX^{k-1}\in\FF_{q^{M/k}}[X].
\]
Now evaluate $p(X)$ in $n$ elements $\alpha_1,\ldots,\alpha_n \in\FF_{q^{M/k}}^*$,
to get the codeword
\[
(p(\alpha_1),\ldots,p(\alpha_n)),~k<n\leq q^{M/k}-1.
\]
\begin{ex}\label{ex:exRS23} \rm
Suppose that $q=2$, so that the size of the object is measured in bits.
Take the 4 bit long object $\ov=(o_1,o_2,o_3,o_4)$, and create $k=2$ fragments: $\ov_1 = (o_1,o_2)\in \FF_4$,
$\ov_2= (o_3,o_4)\in \FF_4$. We use a $(3,2)$ Reed-Solomon code over
$\FF_4$, to store the file in 3 nodes.
Recall that $\FF_4=\{ (a_0,a_1),~a_0,a_1\in \FF_2\}=
\{ a_0+a_1 w,~a_0,a_1\in \FF_2\}$
where $w^2 = w+1$. Thus we can alternatively represent each fragment as:
$\ov_1= o_1+o_2w\in \FF_4$, $\ov_2= o_3+o_4w\in \FF_4$.
The encoding is done by first mapping the two fragments into a polynomial
$p(X)\in\FF_4[X]$:
\[
p(X)=(o_1+o_2w)+(o_3+o_4w)X,
\]
and then evaluating $p(X)$ into the three non-zero elements of $\FF_4$, to
get a codeword of length 3:
\[
(p(1),p(w),p(w+1))
\]
where $p(1)=o_1+o_3+w(o_2+o_4)$, $p(w)=o_1+o_4+w(o_2+o_3+o_4)$,
$p(w^2)=o_1+o_3+o_4+w(o_2+o_3)$,
so that each node gets two bits to store:
$(o_1+o_3,o_2+o_4)$ at node 1, $(o_1+o_4,o_2+o_3+o_4)$ at node 2,
$(o_1+o_3+o_4,o_2+o_3)$ at node 3.
\end{ex}

\begin{defn}
We call {\em homomorphic self-repairing code}, denoted by $\hsrc(n,k)$, the code
obtained by evaluating the polynomial
\begin{equation}\label{eq:linpol}
p(X)=\sum_{i=0}^{k-1}p_iX^{q^i} \in \FF_{q^{M/k}}[X]
\end{equation}
in $n$ non-zero values $\alpha_1,\ldots,\alpha_n$ of
$\FF_{q^{M/k}}$ to get an $n$-dimensional codeword
\[
(p(\alpha_1),\ldots,p(\alpha_n)),
\]
where $p_i=\ov_{i+1}$, $i=0,\ldots,k-1$ and each $p(\alpha_i)$ is given to node $i$ for storage.
\end{defn}

In particular, we need the code parameters $(n,k)$ to satisfy
\begin{equation}\label{eq:boundn}
k<n\leq q^{M/k}-1.
\end{equation}

The analysis that follows refers to this family of self-repairing codes.

%*************************************************************************%
\subsection{Self-repair}
\label{subsec:src}

Since we work over finite fields that contain $\FF_2$, recall that all
operations are done in characteristic 2, that is, scalar operations are performed modulo 2.
Let $a,b \in \FF_{q^m}$, for some $m\geq 1$ and $q=2^t$. Then we have that
$(a+b)^2=a^2+2ab+b^2=a^2+b^2$ since $2ab\equiv 0 \mod 2$, and consequently
\begin{equation}\label{eq:Frob}
(a+b)^{2^i}=\sum_{j=0}^{2^i} {2^i\choose j} a^j b^{2^i-j}= a^{2^i}+b^{2^i},~i\geq 1.
\end{equation}
%Recall the definition of a linearized polynomial.
\begin{defn}
A {\em linearized polynomial} $p(X)$ over $\FF_Q$, $Q=q^m$, has the form
\[
p(X)=\sum_{i=0}^{k-1}p_iX^{Q^i},~p_i\in\FF_Q.
\]
\end{defn}
More generally, one can consider a polynomial $p(X)$ over $\FF_Q$, $Q=q^m$, of the form:
\[
p(X)=\sum_{i=0}^{k-1}p_iX^{s^i},~p_i\in\FF_Q,
\]
where $s=q^l$, $1\leq l\leq m$ ($l=m$ makes $p(X)$ a linearized polynomial).
These polynomials share the following useful property:
\begin{lem}
Let $a,b \in \FF_{q^m}$ and let $p(X)$ be the polynomial given by
$p(X)=\sum_{i=0}^{k-1}p_iX^{s^i}$, $s=q^l$, $m \geq l \geq 1$.
We have
\[
p(ua+vb)=up(a)+vp(b),~u,v\in \FF_s.
\]
\end{lem}
\begin{IEEEproof}
If we evaluate $p(X)$ in $ua+vb$, we get
\[
p(ua+vb)
= \sum_{i=0}^{k-1}p_i(ua+vb)^{s^i} \\
= \sum_{i=0}^{k-1}p_i((ua)^{s^i}+(vb)^{s^i})
\]
by (\ref{eq:Frob}), and
\[
p(ua+vb)
= \sum_{i=0}^{k-1}p_i(ua^{s^i}+vb^{s^i})\\
= u\sum_{i=0}^{k-1}p_ia^{q^i}+ v\sum_{i=0}^{k-1}p_ib^{q^i}
\]
using the property that $u^s=u$ for $u\in\FF_s$.
\end{IEEEproof}

We now define a weakly linearized polynomial as
\begin{defn}
A {\em weakly linearized polynomial} $p(X)$ over $\FF_Q$, $Q=q^m$, has the form
\[
p(X)=\sum_{i=0}^{k-1}p_iX^{q^i},~p_i\in\FF_Q.
\]
\end{defn}
We chose the name weakly linearized polynomial, since we only retain the $\FF_q$-linearity,
namely:
\begin{cor}\label{lem:sr}
Let $a,b \in \FF_{q^m}$ and let $p(X)$ be a weakly linearized polynomial
given by $p(X)=\sum_{i=0}^{k-1}p_iX^{q^i}$.
We have
\begin{equation}\label{eq:hom}
p(ua+vb)=up(a)+vp(b),~u,v\in\FF_q.
\end{equation}
In particular
\begin{equation}\label{eq:hom2}
p(a+b)=p(a)+p(b).
\end{equation}
\end{cor}

It is the choice of a weakly linearized polynomial
in (\ref{eq:linpol}) that enables self-repair.

\begin{ex}\label{ex:poly}\rm
Consider the polynomial
\[
p(X)=p_0X+p_1X^2+p_2X^4 \in \FF_8[X].
\]
We have (see Table \ref{table:ff} for $\FF_8$ arithmetic)
\[
p(w+1)=p_0(w+1)+p_1(w^2+1)+p_2(w^4+1)=p(w)+p(1)
\]
however
\[
p(w^2) \neq p(w)^2
\]
since $p_i^2 = p_i$ if and only if $p_i\in \FF_2$.
\end{ex}

A codeword from $\hsrc(n,k)$ is then of the form $(p(\alpha_1),\ldots,p(\alpha_n))$, where
$p(X)$ is a weakly linearized polynomial.
Since $\FF_{q^{M/k}}$ contains a $\FF_q$-basis $B=\{b_1,\ldots,b_{M/k}\}$, the $\alpha_i$,
$i=1,\ldots,n$, can be expressed as $\FF_q$-linear combinations of the basis
elements, and we have from Lemma \ref{lem:sr} that
\[
\alpha_i=\sum_{j=1}^{M/k}\alpha_{ij}b_j,~\alpha_{ij}\in\FF_q
\Rightarrow
p(\alpha_i)=\sum_{j=1}^{M/k}\alpha_{ij}p(b_j).
\]
In words, that means that if $p$ is evaluated in the elements of the basis $B$ (or any other basis), then any encoded fragment $p(\alpha_i)$ can be obtained as a linear combination of other encoded fragments.

{\bf Controlling the amount of self-repaired redundancy.}
The amount of redundancy allowing self-repair introduced in the coding scheme can be controlled through two mechanisms:
\begin{enumerate}
\item
Firstly, given $k$ fragments, there are different values
of $n$, and different choices of $\{\alpha_1,\ldots,\alpha_n\}$
that can be chosen to define a self-repairing code. Let us denote by $\nmax$ the maximum value that $n$ can take, namely $\nmax=q^{M/k}-1$. By choosing
the set of $\alpha_i$ to form a subspace of $\FF_{\nmax}$, we can reduce the redundancy while maintaining a particularly nice symmetric structure of the code.
In the extreme case where $\alpha_1,\ldots,\alpha_n$ are
contained in $B$, the code has no self-repairing property, and is in fact a MDS code. Thus SRC can be tuned to provide the desired amount of redundancy, from MDS and no self-repair, to the maximal amount of self-repair with $\nmax$.
\item
As seen in Lemma \ref{lem:sr}, the power $s$ of $X^{s^i}$ in the weakly linearized polynomial $p(X)$ determines the $\FF_s$-linearity of $p(X)$. Consequently, the bigger $s$, the more redundancy since from (\ref{eq:hom})
\[
p(ua+vb)=up(a)+vp(b),~u,v\in\FF_q,
\]
meaning that the encoded fragment $p(ua+vb)$ can be repaired by contacting two nodes
$p(a),p(b)$ in as many ways as there are ways for writing $ua+vb$, namely $(q-1)^2$:
\[
ua+vb=u(u')^{-1}u'a+v(v')^{-1}v'b,~u',v'\neq 0,~u',v'\in\FF_q.
\]
In the particular case where $s=2$, we obtain from (\ref{eq:hom2}) an XOR-like
structure
\[
p(a+b)=p(a)+p(b).
\]
However, it is worth remarking that though the encoded fragments
can be thus obtained as XORs of each other, each fragment is actually containing
information about all the different fragments, which is very different than
a simple XOR of the data itself. In particular, HSRC is not a systematic code.
The implications of lack of systematic property will be discussed in Subsection \ref{subsec:syst}.
\end{enumerate}

{\bf Computational complexity of self-repair.}
In terms of computational complexity, the case $s=2$ implies that the cost of a block reconstruction is that of some XORs
(one in the most favorable case, when two terms are enough to reconstruct
a block, up to $k-1$ in the worst case), independently of $q$, since if $q=2^t$, the
addition in $\FF_q$ is done by addition modulo 2 componentwise.
The cost increases if one would like to exploit the $\FF_q$-linearity.
Indeed, repairing through
\[
p(ua+vb)=up(a)+vp(b),~u,v\in\FF_q,
\]
further requires two multiplications in $\FF_q$.

%*************************************************************************%
\subsection{Decoding}

That decoding is possible is guaranteed by either Lagrange interpolation, or
by considering a system of linear equations, assuming that
\begin{equation}\label{eq:boundk}
k\leq {M/k},
\end{equation}
as detailed below.

{\bf Lagrange interpolation.}
Given $k$ fragments $p(\alpha_{i_1}),\ldots,p(\alpha_{i_k})$ such that $\alpha_{i_1},\ldots,\alpha_{i_k}$ are linearly independent, the node that wants to reconstruct the file computes $q^k-1$ linear combinations of the $k$ fragments, which
gives, thanks to the homomorphic property (\ref{eq:hom}), $q^k-1$ points in which $p$ is evaluated.
Lagrange interpolation guarantees that it is enough to have $q^{k-1}+1$ points
(which we have, since $q^k-1 \geq q^{k-1}+1$ for $k\geq 2$) to reconstruct
uniquely the polynomial $p$ and thus the object. This requires
\[
q^{k-1}+1 \leq q^{M/k}-1,
\]
namely there must be enough points in which to evaluate the polynomial, which holds
subject to (\ref{eq:boundk}):
\[
q^k\leq q^{M/k} \Rightarrow  q^{k-1}+1 \leq q^k -1\leq q^{M/k}-1.
\]

{\bf Solving a system of linear equations.}
Alternatively, one can consider decoding as solving a system of linear equations.
Given $k$ linearly independent fragments, say $p(\alpha_{i_1}),\ldots,p(\alpha_{i_k})$,
we can write
\[
\left(
\begin{array}{ccccc}
\alpha_{i_1} & \alpha_{i_1}^{q}&\alpha_{i_1}^{q^2} & \hdots & \alpha_{i_1}^{q^{k-1}}\\
\alpha_{i_2} & \alpha_{i_2}^{q}&\alpha_{i_2}^{q^2} & \hdots & \alpha_{i_2}^{q^{k-1}}\\
\vdots  &  & &  \vdots \\
\alpha_{i_k} & \alpha_{i_k}^q&\alpha_{i_k}^{q^2} & \hdots & \alpha_{i_k}^{q^{k-1}}\\
\end{array}
\right)
\left(\!\!
\begin{array}{c}
p_0 \\
p_1 \\
\vdots \\
p_{k-1}
\end{array}
\!\!\right)
=
\left(\!\!
\begin{array}{c}
p(\alpha_{i_1}) \\
p(\alpha_{i_2}) \\
\vdots \\
p(\alpha_{i_k})
\end{array}
\!\!\right),
\]
and the problem of recovering the object reduces to solving the above system of linear equations.
Note that since $\FF_{q^{M/k}}$ is a vector space of dimension $M/k$ over $\FF_q$, condition
(\ref{eq:boundk}) is needed to guarantee that there exist $k$ linearly independent fragments.

%*************************************************************************%
\subsection{Worked out examples}

Let us first illustrate the choices of the code parameters $(n,k)$, before detailing some
code constructions.

\begin{table}
\begin{center}
\begin{tabular}{|cccl||cccl|}
\hline
$k$ & $M$ & $\nmax$ & $n$       & $k$ & $M$ & $\nmax$ & $n$ \\
\hline
2   & 6   & 7 & 3               & 2   & 6   & $8^3-1$ & $8^2-1$ \\
    & 8   & 15 & 3,~7           &     & 8   & $8^4-1$ & $8^3-1$,$8^2-1$ \\
    & 10 &  31 & 3,~7,~15       &     & 10  & $8^5-1$ & $8^4-1,\ldots,8^2-1$ \\
    & 12 &  63 & 3,~7,~15,~31   &     & 12  & $8^6-1$ & $8^5-1,\ldots,8^2-1$\\
    \hline
3   & 9 & 7 &                   & 3   & 9   & $8^3-1$ & $8^2-1$ \\
    & 12 & 15 & 7               &     & 12  & $8^4-1$ & $8^3-1,8^2-1$\\
    & 15 & 31 & 7,~15           &     & 15  & $8^5-1$ & $8^4-1,\ldots,8^2-1$\\
    \hline
4  & 16 & 15 & 7                & 4   & 16  & $8^4-1$ & $8^3-1,8^2-1$\\
5  & 30 & 63 & 7,~15,~31        & 5   & 30  & $8^6-1$ & $8^5-1,\ldots,8^2-1$\\
6  & 42 & 127 & 7,~15,~31,~63   & 6   & 42  & $8^7-1$ & $8^6-1,\ldots,8^2-1$\\
\hline
    \end{tabular}
\caption{Examples of small code parameters for $q=2$ (on the left) and $q=8$ (on the right).}
\label{tab:param}
\end{center}
\end{table}

We recall that the parameters $(n,k)$ of an $\hsrc(n,k)$ code must satisfy conditions
(\ref{eq:boundn}) and (\ref{eq:boundk}):
\[
k<n\leq q^{M/k}-1,~k\leq M/k.
\]
Thus for any choice of $k$:
\begin {enumerate}
\item
pick any $M$ which is a multiple of $k$ (zero padding can be used to remove the constraint on the real size of the object),
\item
define
\[
\nmax = q^{M/k} -1,
\]
\item
pick any $n$ such that
\[
n>k,~n < \nmax
\]
which is a power of $q$ minus 1 (this last condition is not completely necessary but ensures symmetry as already mentioned above).
\end{enumerate}
Some examples of small parameters $(n,k)$ are given in Table \ref{tab:param} for $q=2$ and $q=8$.

\begin{ex}\label{ex:complete}\rm
Take a data file $\ov=(o_1,\ldots,o_{12})$ of $M=12$ bits ($q=2$), and choose
$k=3$ fragments. We have that $M/k=4$, which satisfies (\ref{eq:boundk}),
that is $k=3 \leq M/k= 4$.

The file $\ov$  is cut into 3 fragments $\ov_1=(o_1,\ldots,o_4)$,
$\ov_2=(o_5,\ldots,o_8)$, $\ov_3=(o_9,\ldots,o_{12}) \in \FF_{2^4}$.
Let $w$ be a generator of the multiplicative group $\FF_{2^4}^*$, such
that $w^4=w+1$. The polynomial used for the encoding is
\[
p(X)=\sum_{i=1}^4o_iw^iX+\sum_{i=1}^4o_{i+4}w^iX^2+\sum_{i=1}^4o_{i+8}w^iX^4.
\]
The $n$-dimensional codeword is obtained by evaluating $p(X)$ in $n$ elements
of $\FF_{2^4}$, $n\leq 15=\nmax$ by (\ref{eq:boundn}).

For $n=4$, if we evaluate $p(X)$ in $w^i$, $i=0,1,2,3$, then the 4
encoded fragments $p(1),p(w),p(w^2),p(w^3)$ are linearly independent and there
is no self-repair possible.

Now for $n=7$, and say, $1,w,w^2,w^4,w^5,w^8,w^{10}$, we get:
\[
(p(1),p(w),p(w^2),p(w^4),p(w^5),p(w^8),p(w^{10})).
\]
Suppose node 5 which stores $p(w^5)$ goes offline. A new comer can get
$p(w^5)$ by asking for $p(w^2)$ and $p(w)$, since
\[
p(w^5)=p(w^2+w)=p(w^2)+p(w).
\]
Table \ref{tab:enumerate} shows other examples of missing fragments and which
pairs can reconstruct them, depending on if 1, 2, or 3 fragments are missing
at the same time.

\begin{table}
\begin{tabular}{|c|c|}
  \hline
  missing &pairs to reconstruct missing fragment(s)\\fragment(s)&\\
  \hline
  $p(1)$   & $(p(w),p(w^4))$;$(p(w^2),p(w^8))$;$(p(w^5),p(w^{10}))$\\
  $p(w)$   & $(p(1),p(w^4))$;$(p(w^2),p(w^5))$;$(p(w^8),p(w^{10}))$\\
  $p(w^2)$ & $(p(1),p(w^8))$;$(p(w),p(w^5))$;$(p(w^4),p(w^{10}))$\\
  \hline
  $p(1)$ and & $(p(w^2),p(w^8))$ or $(p(w^5),p(w^{10}))$ for $p(1)$\\
  $p(w)$     & $(p(w^8),p(w^{10}))$ or $(p(w^2),p(w^{5}))$ for $p(w)$\\
  \hline
  $p(1)$ and & $(p(w^5),p(w^{10}))$ for $p(1)$\\
  $p(w)$ and & $(p(w^8),p(w^{10}))$ for $p(w)$\\
  $p(w^2)$   & $(p(w^4),p(w^{10}))$ for $p(w^2)$  \\
  \hline
\end{tabular}
\vspace{1mm}
\caption{Ways of reconstructing missing fragment(s) in Example
\ref{ex:complete}}\vspace{-6mm}
\label{tab:enumerate}
\end{table}
As for decoding, since $p(X)$ is of degree 5, a node that wants to recover the
data needs $k=3$ linearly independent fragments, say $p(w),p(w^2),p(w^3)$, out
of which it can generate $p(aw+bw^2+cw^3)$, $a,b,c \in \{0,1\}$.
Out of the $7$ non-zero coefficients, 5 of them are enough to recover $p$.
\end{ex}

\begin{ex}\label{ex:completeF8}\rm
Take now a data file $\ov=(o_1,\ldots,o_{16})$ of $M=16$ bytes (that is $q=8$), and choose
$k=4$ fragments. We have that $M/k=4$, which satisfies (\ref{eq:boundk}),
that is $k \leq M/k$.

The file $\ov$  is cut into 4 fragments $\ov_1=(o_1,\ldots,o_4)$,
$\ov_2=(o_5,\ldots,o_8)$, $\ov_3=(o_9,\ldots,o_{12})$, $\ov_4= (o_{13},\ldots,o_{16})\in \FF_{8^4}$.
Let $w$ be a generator of the multiplicative group of $\FF_{8}$, and $\nu$ be a generator of the
multiplicative group of $\FF_{8^4}/\FF_8$ such
that $\nu^4=\nu^3+w$. The polynomial used for the encoding can be either
\[
p(X)=\sum_{i=1}^4o_i\nu^iX+\sum_{i=1}^4o_{i+4}\nu^iX^8
\]
\[
+\sum_{i=1}^4o_{i+8}w\nu^iX^{64}+\sum_{i=1}^4o_{i+12}w\nu^iX^{512},
\]
or
\[
p'(X)=\sum_{i=1}^4o_i\nu^iX+\sum_{i=1}^4o_{i+4}\nu^iX^2
\]
\[
+\sum_{i=1}^4o_{i+8}w\nu^iX^4+\sum_{i=1}^4o_{i+12}w\nu^iX^8.
\]
The $n$-dimensional codeword is obtained by evaluating $p(X)$ in $n$ elements
of $\FF_{8^4}$, $n\leq 8^4-1=\nmax$ by (\ref{eq:boundn}).

Now for $n=63$, and say, $\{ u + \nu v,~u,v \in\FF_8,~(u,v)\neq (0,0) \}$, we get:
\[
\{ p(u + \nu v),~u,v \in\FF_8,~(u,v)\neq (0,0) \},
\]
respectively
\[
\{ p'(u + \nu v),~u,v \in\FF_8,~(u,v)\neq (0,0) \}.
\]
Let us give an example of repair in both cases. Let us start with $p'(X)$.
Suppose the node which stores $p'(w+\nu)$ goes offline. A new comer can get
$p'(w+\nu)$ by asking for $p'(w)$ and $p'(\nu)$, since
\[
p'(w+\nu)=p'(w)+p'(\nu).
\]
If we are instead using $p(X)$, when the node which stores $p(w+\nu)$ goes offline, then
for any choice of $u,v\neq0,~u,v\in \FF_8$, a new comer can ask $p(uw)$ and $p(v\nu)$, and then compute
\[
u^{-1}p(uw)+v^{-1}p(v\nu)=p(w)+p(\nu).
\]

\end{ex}

%\textbf{Note for myself (AD): Need to edit the real life example to smoothen rough edges.}
\begin{ex}\label{ex:real}\rm
The HSRC codes described above implicitly assume a specific, fixed input size, determined by the choices of $n,k,q$ on which the coding is to be performed. This is also the case for many other coding schemes such as Reed-Solomon codes. In real life, data objects may however come in an arbitrary size. Two heuristics deal with the consequent constraints - namely, zero padding (for object which is too small), and slicing (for a large object).

For instance, consider as object $\ov$ a file of $5$MB = $5\cdot 2^{10}$ KB = $5\cdot 2^{20}$ B, so that $M=5\cdot 2^{20}$ for $q=8$.
We cut $\ov$ into $2^8$ slices $\sv_1,\ldots,\sv_{256}$, each slice has size $M' =5\cdot 2^{12}$B=20 KB.
We now encode each $\sv_i$ into a codeword $\xv_i=(x_{i1},\ldots,x_{i,511})$, using a $(511,k')$ codes, with $k'=5\cdot 2^4=80$, so that the size of the encoded blocks is $M'/k'=2^8=256B$.
Then the $j$th node stores the $j$th encoded fragment $x_{ij}$ $j=1,\ldots,511$  for all the slices $i=1,\ldots,256$.
To get a point of comparison of the parameter values $(n,k)$ used in this example, we note that Wuala uses a (517,100) code, while (255,223) Reed Solomon codes (in bytes) are also widely used.

A final observation we want to make here is that, in the following section, when we analyse the static resilience of a code, it determines the availability for one slice of the object, rather than the object itself. However, placing the encoded fragments of all the slices in a common pool of storage nodes - which is also practical in terms of managing meta-information - leads to the same availability of the object, as for the individual slices. Hence we do not distinguish the two in the rest of this paper, and consider that the code could be applied on any object, independently of its size. It has to be noted here that if the encoded fragments of different slices were to be placed among different set of nodes, this would however not hold true. This however is an issue we will not delve into any further, and is also not usually practiced due to practical system design considerations.

%{\bf CHECK IF RS DOES SLICING+explain why placement allows same static resilience analyisi}
\end{ex}

%*************************************************************************%
%
% STATIC RESILIENCE ANALYSIS
%
%*************************************************************************%
\section{Static Resilience Analysis}
\label{sec:static}

The rest of the paper is dedicated to the analysis of the proposed homomorphic
self-repairing codes.
\emph{Static resilience} of a distributed storage system is defined as the probability that an object, once stored in the system, will continue to stay available without any further maintenance, even when a certain fraction of individual member nodes of the distributed system become unavailable.
We start the evaluation of the proposed scheme with a static resilience analysis, where we study how a stored object can be recovered using HSRCs, compared with traditional erasure codes, prior to considering the maintenance process, which will be done in Section \ref{sec:dynamic}.

Let $\probup$ be the probability that any specific node is available. Then, under the assumptions that node availability is $i.i.d$, and no two fragments of the same object are placed on any same node, we can consider that the availability of any fragment is also $i.i.d$ with probability
$\probup$.

%**************************************************************************%
\subsection{A network matrix representation}

Recall that using the above coding strategy, an object $\ov$ of length $M$ is decomposed
into $k$ fragments of length $M/k$:
\[
\ov=(\ov_1,\ldots,\ov_k),~\ov_i\in\FF_{q^{M/k}},
\]
which are further encoded into $n$ fragments of same length:
\[
\xv=(\xv_1,\ldots,\xv_n),~\xv_i\in\FF_{q^{M/k}},
\]
each of the encoded fragment $\xv_i=p(\alpha_i)$ is given to a node to be stored.
We thus have $n$ nodes each possessing a $q$-ary vector of length $M/k$,
corresponding to a system of $n$ linear equations
\[
\left(
\begin{array}{ccccc}
\alpha_1 & \alpha_{1}^q&\alpha_{1}^{q^2} & \hdots & \alpha_{1}^{q^{k-1}}\\
\alpha_2 & \alpha_{2}^q&\alpha_{2}^{q^2} & \hdots & \alpha_{2}^{q^{k-1}}\\
\vdots  &  & &  \vdots \\
\alpha_n & \alpha_{n}^q&\alpha_{n}^{q^2} & \hdots & \alpha_{n}^{q^{k-1}}\\
\end{array}
\right)
\left(\!\!
\begin{array}{c}
p_0 \\
p_1 \\
\vdots \\
p_{k-1}
\end{array}
\!\!\right)
=
\left(\!\!
\begin{array}{c}
p(\alpha_{1}) \\
p(\alpha_{2}) \\
\vdots \\
p(\alpha_{n})
\end{array}
\!\!\right).
\]
If three rows are $\FF_q$-linearly dependent, say rows $1,2$ and $3$, then
\[
u(\alpha_1,\alpha_{1}^q,\alpha_{1}^{q^2},\hdots,\alpha_{1}^{q^{k-1}})+
v(\alpha_2,\alpha_{2}^q,\alpha_{2}^{q^2}, \hdots,\alpha_{2}^{q^{k-1}})
\]
\[
=
(\alpha_3,\alpha_{3}^q,\alpha_{3}^{q^2}, \hdots,\alpha_{3}^{q^{k-1}}),~u,v\in\FF_q,
\]
which can be rewritten as
\[
(u\alpha_1+v\alpha_2,u\alpha_{1}^q+v\alpha_{2}^q,u\alpha_{1}^{q^2}+v\alpha_{2}^{q^2},\hdots,u\alpha_{1}^{q^{k-1}}+v\alpha_{2}^{q^{k-1}})
\]
\[
=
(u\alpha_1+v\alpha_2,(u\alpha_{1}+v\alpha_{2})^q,(u\alpha_{1}+v\alpha_{2})^{q^2}\!\!,\hdots,(u\alpha_{1}+v\alpha_{2})^{q^{k-1}})
\]
\[
=
(\alpha_3,\alpha_{3}^q,\alpha_{3}^{q^2}, \hdots,\alpha_{3}^{q^{k-1}}),~u,v\in\FF_q
\iff u\alpha_1+v\alpha_2=\alpha_3.
\]
Thus to understand the linear dependencies among the fragments owed by each of the $n$ nodes, one can associate to
the $i$th node the value $\alpha_i$. Once all the $\alpha_i$ are written in a $\FF_q$-basis, they can be represented as an
$n\times M/k$ $q$-ary matrix
\begin{equation}\label{eq:MM}
\MM=
\left(
\begin{array}{c}
\alpha_1 \\
\vdots \\
\alpha_n
\end{array}
\right)
=
\left(
\begin{array}{ccc}
\alpha_{1,1} & \ldots & \alpha_{1,M/k}\\
\vdots &  &\vdots \\
\alpha_{n,1}& \ldots & \alpha_{n,M/k}
\end{array}
\right)
\end{equation}
with $\alpha_{i,j}\in\FF_q$.

\begin{ex}\label{ex:M}\rm
In Example \ref{ex:complete}, by fixing as $\FF_2$-basis $\{1,w,w^2,w^3\}$, we have for $n=4$ that $\MM=I_4$, the 4-dimensional identity matrix, since $\alpha_i=w^i,~i=0,\ldots,3$, while for $n=7$, it is
\[
\MM^T=
\left(
\begin{array}{ccccccc}
1&0&0&1&0&1&1\\
0&1&0&1&1&0&1\\
0&0&1&0&1&1&1\\
0&0&0&0&0&0&0\\
\end{array}
\right),
\]
corresponding to $1,w,w^2,w^4,w^5,w^8,w^{10}$.
Now in Example \ref{ex:completeF8}, by fixing as $\FF_8$-basis $\{1,\nu,\nu^2,\nu^3\}$, we have
\[
\MM^T=
\left(
\begin{array}{c}
u \\
v \\
0 \\
0
\end{array}
\right),~u,v\in\FF_8,~(u,v)\neq 0.
\]
\end{ex}
Thus unavailability of a random node is equivalent to losing one linear equation, or a random row of the matrix $\mathbb{M}$. If multiple random nodes (say $n-x$) become unavailable, then the remaining $x$ nodes provide $x$ encoded fragments,
which can be represented by a $x\times M/k$ sub-matrix $\mathbb{M}_x$
of $\MM$. For any given combination of such $x$ available encoded fragments, the original object can still be reconstructed if we can obtain at least $k$ linearly independent rows of $\MM_x$. This is equivalent to say that the object can be reconstructed if the rank of the matrix
$\mathbb{M}_x$ is larger than or equal to $k$.

In the case the polynomial $p(X)=\sum_{i=0}^{k-1}p_iX^{2^i},~p_i\in\FF_q$ is chosen for encoding,
the corresponding system of $n$ linear equations is slightly different
\[
\left(
\begin{array}{ccccc}
\alpha_1 & \alpha_{1}^2&\alpha_{1}^{2^2} & \hdots & \alpha_{1}^{2^{k-1}}\\
\alpha_2 & \alpha_{2}^2&\alpha_{2}^{2^2} & \hdots & \alpha_{2}^{2^{k-1}}\\
\vdots  &  & &  \vdots \\
\alpha_n & \alpha_{n}^2&\alpha_{n}^{2^2} & \hdots & \alpha_{n}^{2^{k-1}}\\
\end{array}
\right)
\left(\!\!
\begin{array}{c}
p_0 \\
p_1 \\
\vdots \\
p_{k-1}
\end{array}
\!\!\right)
=
\left(\!\!
\begin{array}{c}
p(\alpha_{1}) \\
p(\alpha_{2}) \\
\vdots \\
p(\alpha_{n})
\end{array}
\!\!\right),
\]
so that now, though it is still true that if three rows are $\FF_q$-linearly dependent, say rows $1,2$ and $3$, then
\[
u(\alpha_1,\alpha_{1}^2,\alpha_{1}^{2^2},\hdots,\alpha_{1}^{2^{k-1}})+
v(\alpha_2,\alpha_{2}^2,\alpha_{2}^{2^2}, \hdots,\alpha_{2}^{2^{k-1}})
\]
\[
=
(\alpha_3,\alpha_{3}^2,\alpha_{3}^{2^2}, \hdots,\alpha_{3}^{2^{k-1}}),~u,v\in\FF_q,
\]
it is not true anymore that
\[
(u\alpha_1+v\alpha_2,u\alpha_{1}^2+v\alpha_{2}^2,u\alpha_{1}^{2^2}+v\alpha_{2}^{2^2},\hdots,u\alpha_{1}^{2^{k-1}}+v\alpha_{2}^{2^{k-1}})
\]
\[
=
(u\alpha_1+v\alpha_2,(u\alpha_{1}+v\alpha_{2})^2,(u\alpha_{1}+v\alpha_{2})^{2^2}\!\!,\hdots,(u\alpha_{1}+v\alpha_{2})^{2^{k-1}})
\]
since $u^2=u$, resp. $v^2=v$ holds if and only if $u,v \in \FF_2$.
In this case, we have to analyze the matrix
\begin{equation}
\left(
\begin{array}{ccccc}
\alpha_1 & \alpha_{1}^2&\alpha_{1}^{2^2} & \hdots & \alpha_{1}^{2^{k-1}}\\
\alpha_2 & \alpha_{2}^2&\alpha_{2}^{2^2} & \hdots & \alpha_{2}^{2^{k-1}}\\
\vdots  &  & &  \vdots \\
\alpha_n & \alpha_{n}^2&\alpha_{n}^{2^2} & \hdots & \alpha_{n}^{2^{k-1}}\\
\end{array}
\right)
\end{equation}
directly.

\begin{ex}\label{ex:fqlin}
Consider again Example \ref{ex:completeF8}, and suppose the two polynomials
$p(X)$ and $p'(X)$ are both evaluated in $\alpha_1=1$, $\alpha_2=\nu$, and $\alpha_3=w+\nu$.
Clearly $w\alpha_1+\alpha_2=\alpha_3$. Consequently, when evaluating the polynomial $p(X)$
in $\alpha_1,\alpha_2,\alpha_3$, we get as part of the system of linear equations the following 3 rows:
\[
(1,1,1,1),(\nu,\nu^8,\nu^{64},\nu^{512})
\]
and
\[
(w+\nu,(w+\nu)^8,(w+\nu)^{64},(w+\nu)^{512})=
\]
\[
(w+\nu,w^8+\nu^8,w^{64}+\nu^{64},w^{512}+\nu^{512}).
\]
Cleary the three rows are linearly dependent. If now instead the polynomial $p'(X)$ is similarly evaluated,
we obtain
\[
(1,1,1,1),(\nu,\nu^2,\nu^{4},\nu^{8})
\]
and
\[
(w+\nu,(w+\nu)^2,(w+\nu)^{4},(w+\nu)^{8}).
\]
This time the dependencies disappear, since
\[
w+\nu^2 \neq (w+\nu)^2=w^2+\nu^2.
\]
\end{ex}

\begin{center}

\begin{figure*}[ht]
 \subfigure[Validation of the static resilience analysis]{
  \includegraphics[scale=0.34]{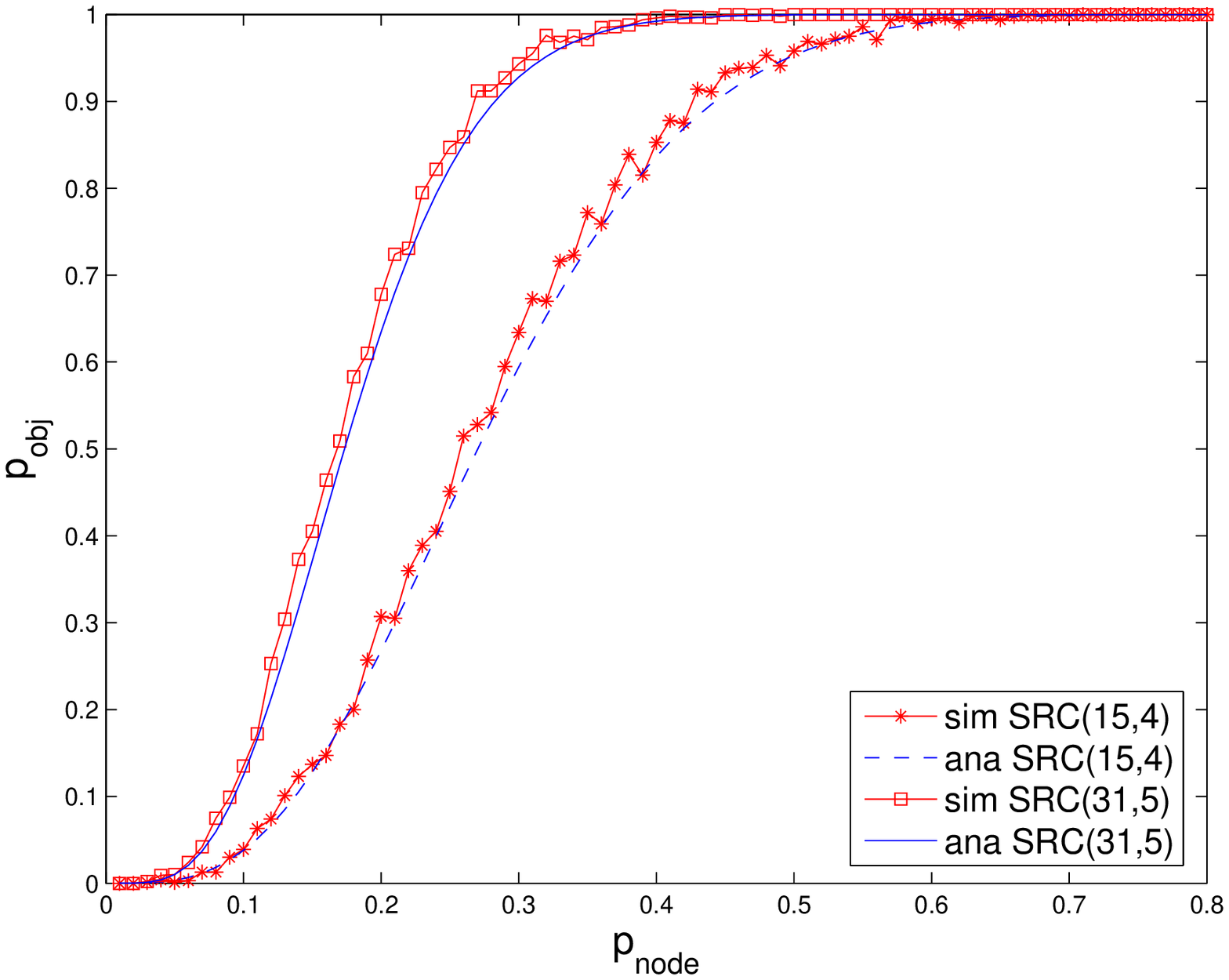}
   \label{fig:simvsana}}
 \subfigure[Comparison of SRC with EC]{
  \includegraphics[scale=0.34]{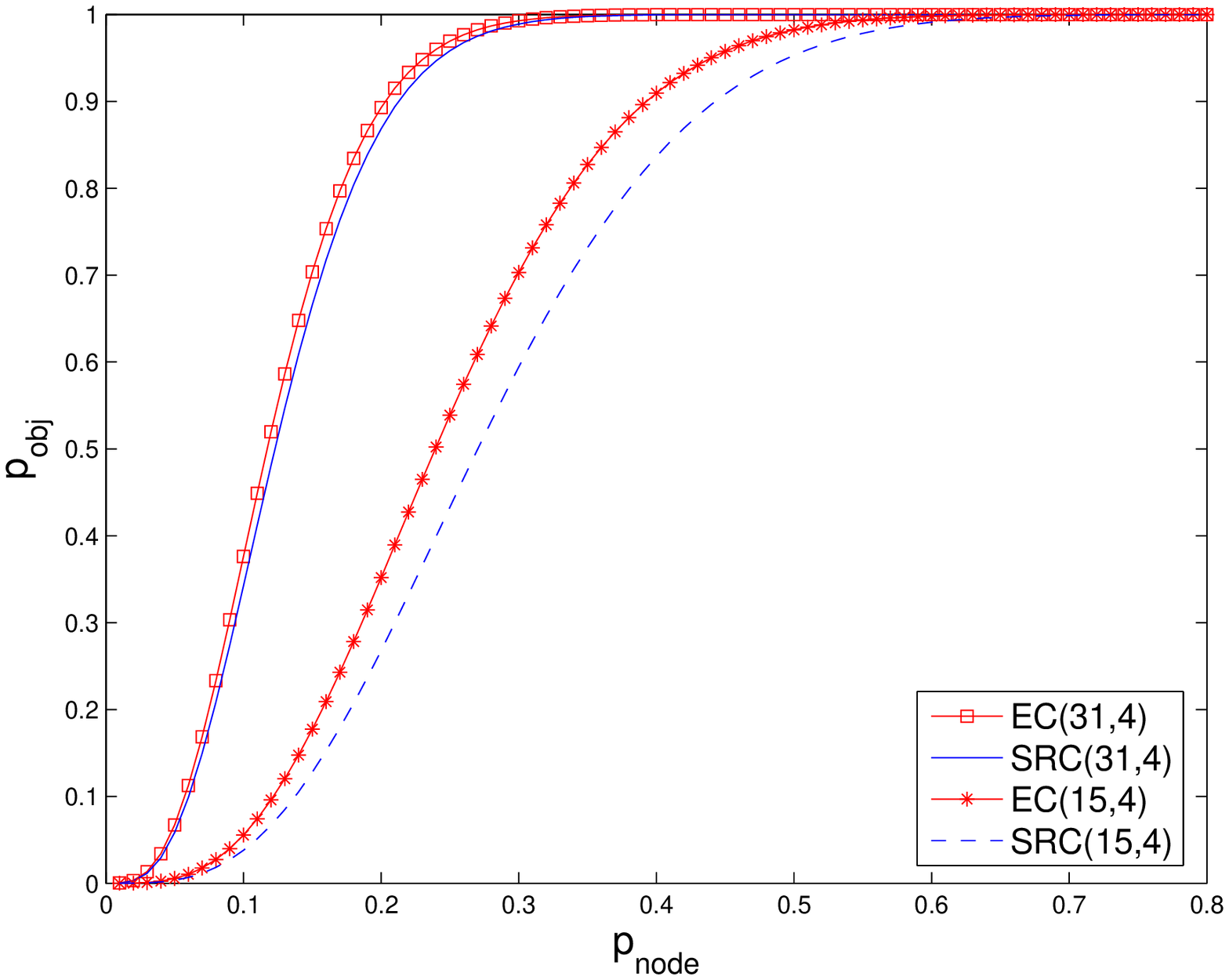}
   \label{fig:ECvsSRCk4}}
 \subfigure[Comparison of SRC with EC]{
  \includegraphics[scale=0.34]{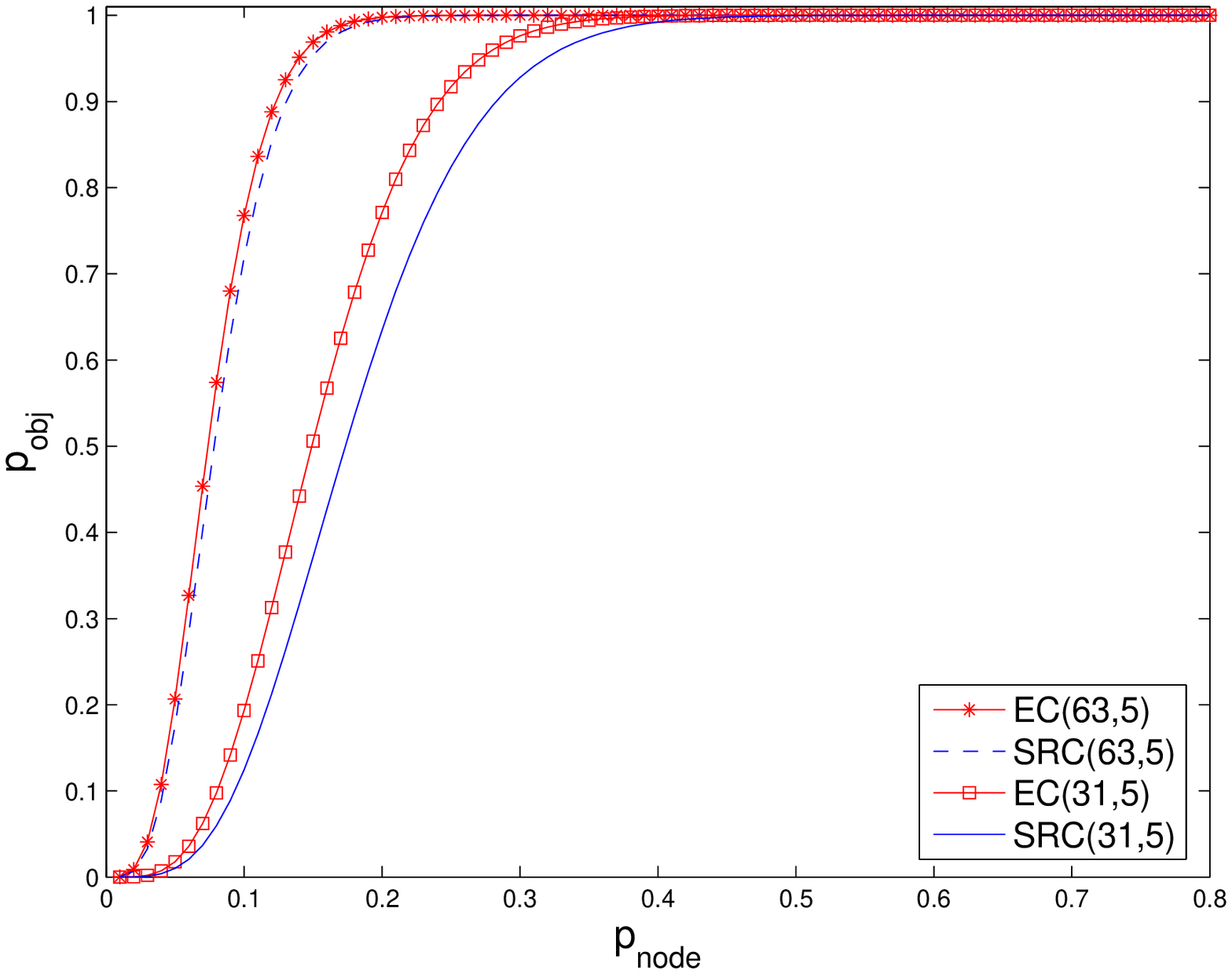}
   \label{fig:ECvsSRCk5}}
 \caption{Static resilience of homomorphic self-repairing codes (\hsrc) for $q=2$: Validation of analysis, and comparison with MDS erasure codes (EC)} \label{fig:staticresilience}
 \vspace{3mm}
\end{figure*}
\end{center}

%*************************************************************************%
\subsection{Probability of object retrieval}

Consider a $(q^d-1) \times d$ $q$-ary matrix for some $d>1$, with distinct rows, no all zero row, and thus rank $d$.
The case of interest for us is $d=M/k$, since $\mathbb{M}$ is an $q^{M/k}-1 \times M/k$ matrix.
If we remove some of the rows uniformly randomly with some probability
$1-\probup$, then we are left with a $x \times d$ sub-matrix - where $x$ is binomially distributed. We define $R(x,d,r)$ as the number of
$x \times d$ sub-matrices with rank $r$, voluntarily including all the possible permutations of the rows in the counting.

\begin{lem}\label{lem:count}
Let $R(x,d,r)$ be the number of $x\times d$ sub-matrices
with rank $r$ of a tall $(q^d-1)\times d$ matrix of rank $d$.
We have that $R(x,d,r)=0$ when (i) $r=0$, (ii) $r>x$, (iii) $r=x$, with $x>d$, or (iv) $r<x$ but $r>d$.
Then, counting row permutations:
\[
R(x,d,r)=
\prod_{i=0}^{r-1}(q^d-q^i)\mbox{ if }r=x,x\leq d,
\]
and for $r< x$ with $r\leq d$:
\[
R(x,d,r)=
R(x-1,d,r-1)(q^d-q^{r-1})+R(x-1,d,r)(q^r-x).
\]
\end{lem}
\begin{proof}
There are no non-trivial matrix with rank $r=0$.
When $r>x$, $r=x$ with $x>d$, or $r<x$ but $r>d$, $R(x,d,r)=0$ since the rank of a matrix cannot be larger than the smallest of its dimensions.

For the case when $r=x$, with $x\leq d$, we deduce $R(x,d,r)$ as follows. To build a matrix $\MM_x$ of rank $x=r$, the first row can be chosen from any of the $q^d-1$ rows in $\MM$, and the second row should not be a multiple of the first row, which gives $q^d-2$ choices. The third row needs to be linearly independent from the first two rows. Since there are $q^2$ linear combinations of the first two rows, which includes the all zero vector which is discarded, we obtain  $q^d-q^2$ choices. In general, the $(i+1)$st row can be chosen from $q^d-q^i$ options that are linearly independent from the $i$ rows that have already been chosen. We thus obtain $R(x,d,r) = \prod_{i=0}^{r-1}(q^d-q^i)$ for $r=x$, $x\leq d$.

For the case where $r< x$ with $r\leq d$, we observe that $x \times d$ matrices of rank $r$ can be inductively obtained by either (I) adding a linearly independent row to a $(x-1) \times d$ matrix of rank $r-1$, or (II) adding a linearly dependent row to a $(x-1) \times d$ matrix of rank $r$. We use this observation to derive the recursive relation
\[
R(x,d,r)= R(x-1,d,r-1)(q^d-q^{r-1})+R(x-1,d,r)(q^r-x),
\]
where $q^d-1-(q^{r-1}-1)$ counts the number of
linearly independent rows that can be added, and $q^r-1-(x-1)$ is on the
contrary the number of linearly dependent rows.
\end{proof}
We now remove the permutations that we counted in the above analysis by introducing
a suitable normalization.
\begin{cor}\label{cor:rho}
Let $\rho(x,d,r)$ be the fraction of sub-matrices of dimension $x \times d$ with rank $r$
out of all possible sub-matrices of the same dimension. Then
\[
\rho(x,d,r)=\frac{R(x,d,r)}{\sum_{j=0}^dR(x,d,j)}=\frac{R(x,d,r)}{C_{x}^{q^d-1}x!}.
\]
In particular
\begin{equation}\label{eq:rhox}
\rho_x(d)=\sum_{r=k}^d \rho(x,d,r)
\end{equation}
is the conditional probability that the stored object can be retrieved by contacting an arbitrary $x$ out of the $n$ storage nodes.
\end{cor}
\begin{proof}
It is enough to notice that there are $C_{x}^{q^d-1}$ ways to choose $x$ rows out of the possible $q^d-1$ options.
The chosen $x$ rows can be ordered in $x!$ permutations.

In particular, when the rank is at least $k$, the object can be retrieved.
\end{proof}

We now put together the above results to compute the probability $\objup$ of an object being recoverable when
using an $\hsrc(n,k)$ code to store a length $M$ object made of $k$ fragments encoded into $n$ fragments each
of length $M/k$.
\begin{cor}
Using an $\hsrc(n,k)$, the probability $\objup$ of recovering the object is
\[
\objup = \sum_{x=k}^n \sum_{r=k}^d \rho(x,d,r) C_{x}^{n} \probup^x (1-\probup)^{n-x},
\]
where $d=\log_q{n+1}$.
\end{cor}
\begin{IEEEproof}
If $n=\nmax=q^{M/k}-1$, we apply Lemma \ref{lem:count} and Corollary
\ref{cor:rho} with $d=M/k$.
If $n=q^i-1$, for some integer $i\leq M/k$ such that $n>k$ (otherwise there is no encoding), then $\MM$ has $M/k-i$ columns which are either all zeros or all ones vectors, as shown on Example \ref{ex:M}.
Thus the number of its sub-matrices of rank $r$ is given by applying Lemma \ref{lem:count} on the matrix obtained
by removing these redundant columns.
\end{IEEEproof}
We validate the analysis with simulations, and as can be observed from Figure \ref{fig:simvsana}, we obtain a precise match.

To conclude this analysis, let us get back to Example \ref{ex:completeF8} and notice that the static resilience analysis
derived above holds for the encoding via the polynomial $p(X)$. It would not be the case if $p'(X)$ were used, in which
case only the $\FF_2$-linear dependencies should be kept.

%*************************************************************************%
\subsection{Comparison with standard erasure codes}
\label{sec:comapreMDS}
While there is marginal deterioration of static resilience using SRC with respect to MDS codes (as compared in Fig. \ref{fig:staticresilience}, and to be discussed soon after), we first elaborate how SRC differs fundamentally from MDS codes by looking at the conditional probability that the stored object can be retrieved by contacting an arbitrary $x$ out of the $n$ storage nodes.

For $(n,k)$ MDS erasure codes, $\rho_x$ is a deterministic and binary value equal to one for $x \geq k$, and zero for smaller $x$. For self-repairing codes, the value is probabilistic. In Fig. \ref{fig:oneminusrho} we show for our toy example $HSRC(31,5)$ the probability that the object can be retrieved by contacting arbitrary $x$ nodes, i.e., $\rho_x$, where the values of $\rho_x$ for $x \geq k$ were computed from (\ref{eq:rhox}).\footnote{$\rho_x$ is zero for $x<k$ for \hsrc also.} This can alternatively be interpreted as the probability that the object is retrievable despite precisely $n-x$ random failures, and only $x$ random storage nodes are available.

In particular, if any five storage nodes are randomly picked, it is likely that the object cannot be reconstructed with a probability  0.5096, while if any seven random nodes are picked, this probability decreases to 0.0757, while, if thirteen or more random nodes are picked, then the object can certainly be reconstructed. In contrast, for MDS codes, the object will be retrievable from the data available at any arbitrary five nodes.

Of-course, this rather marginal sacrifice (we will next compare HSRC's static resilience with MDS erasure codes to demonstrate the marginality) provides HSRC an incredible amount of self-repairing capability. Given that a practical system will carry out repairs rather frequently, and HSRC in fact allows very cheap repairs, a system using HSRC will be more easily and cheaply maintained, and hence be reliable - particularly be avoiding multiple failures to cumulate.

Likewise, for data access and reconstruction, in practice, storage nodes will be accessed in a planned manner, rather than randomly. There are in fact $C_{k}^{n} \rho_k$ (i.e., 83324 for $\hsrc(31,5)$) unique subsets of precisely $k$ storage nodes that allows reconstruction. Hence, in practice, object access overheads will not be different than when using a MDS coding based scheme.

Let us now compare HSRC against standard MDS erasure codes in terms of the effective static resilience. If we use a $(n,k)$ MDS erasure code, then the probability
that the object is recoverable when each individual storage node may fail i.i.d. with probability $1-\probup$ is:
\[
\objup = \sum_{i=k}^{n}C_{i}^n \probup^i (1-\probup)^{n-i}.
\]
Note that MDS codes may not exist for specific arbitrary choice of $n$ and $k$. However, for the sake of fair comparison, this
formula and the following plots are provided as if they were to exist.

\begin{figure}
\includegraphics[scale=0.6]{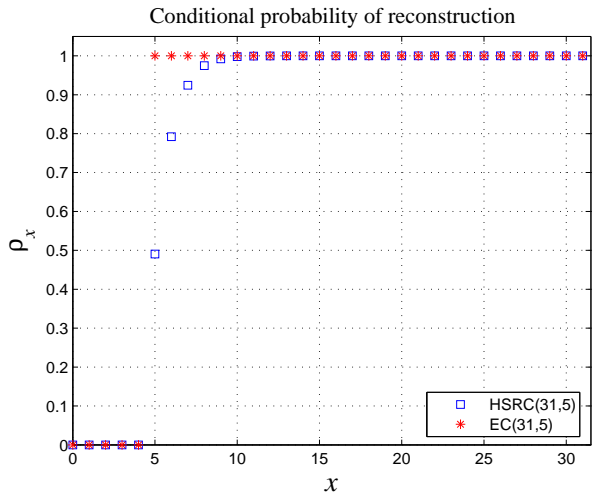}
\caption{Comparison of the probability of reconstruction of object using encoded data in $x$ random storage nodes}
\label{fig:oneminusrho}
\end{figure}

In Figures \ref{fig:ECvsSRCk4} and \ref{fig:ECvsSRCk5}, we compare the static resilience achieved using the proposed homomorphic SRC with that of MDS erasure codes.

In order to achieve the self-repairing property in SRC, it is obvious that it is necessary to introduce extra `redundancy' in its code structure, but we notice from the comparisons that this overhead is in fact marginal. For the same storage overhead $n/k$, the overall static resilience of SRC is only slightly lower than that of EC, and furthermore, for a fixed $k$, as the value of $n$ increases, SRC's static resilience gets very close to that of EC. Furthermore, even for low storage overheads, with relatively high $\probup$, the probability of object availability is indeed 1. In any storage system, there will be a maintenance operation to replenish lost fragments (and hence, the system will operate for high values of $\probup$). We will further see in the next section that SRCs have significantly lower maintenance overheads. These make SRCs a practical coding scheme for networked storage.

%{\bf FIX K AND STATIC RESILIENCE AND COMPUTE N TO GET ANOTHER COMPARISON IN TERMS OF STORAGE COST --- This should be here, and %not later, where you had indicated previously - keyonki that is the section for communication overheads, and this rightfully %belongs with the rest of the discussion on MDS ...}

%*************************************************************************%
%
% DYNAMIC ANALYSIS
%
%*************************************************************************%

\section{Communication overheads of self-repair}
\label{sec:dynamic}

In the previous section we studied the probability of recovering an object if it so happens that only $\probup$ fraction of nodes which had originally stored the encoded fragments continue to remain available, while lost redundancy is yet to be replenished. Such a situation may arise either because a lazy maintenance mechanism (such as, in \cite{TotalRecall}) is applied, which triggers repairs only when redundancy is reduced to certain threshold, or else because of multiple correlated failures before repair operations may be carried out. We will next investigate the communication overheads in such scenarios, emphasizing on those HSRC with an XOR-like structure (that is, retaining $\FF_2$-linearity).
Note that this is really the regime in which we need an analysis, since
in absence of correlated failures, and assuming that an eager repair strategy is applied, whenever one encoded block is detected to be unavailable, it is immediately replenished. The proposed HSRC ensures that this one missing fragment can be replenished by obtaining only two other (appropriate) encoded fragments, thanks to the HSRC subspace structure.

\begin{defn}The \emph{diversity} $\divsrc$ of SRC is defined as the number of mutually exclusive pairs of fragments which can be used to recreate any specific fragment.
\end{defn}

In Example \ref{ex:complete}, it can be seen easily that $\divsrc=3$. Let us assume that $p(w)$ is missing. Any of the three exclusive fragment pairs, namely $((p(1),p(w^4))$; $(p(w^2),p(w^5))$ or $(p(w^8),p(w^{10}))$ may be used to reconstruct $p(w)$. See Table \ref{tab:enumerate} for other examples. In Example \ref{ex:completeF8} where the encoding is done using
$p'(X)$, the diversity is $\divsrc=31$. Indeed, every encoded fragment is of the form $p'(u+\nu v)$, $u,v\in\FF_8$, so that
for every $u'+\nu v'$, $u',v'\in\FF_8$, we have that the pair $(p'(u'+\nu v'),p'((u'+u)+\nu (v'+v)))$ can be used to reconstruct $p'(u+\nu v)$, since $p'(u'+\nu v')+p'((u'+u)+\nu (v'+v))= p'(u+\nu v)$, and the fragment $p'((u'+u)+\nu (v'+v))$ is indeed present in the network since $u',v'\in\FF_8$ and $u,v$ run through every element in $\FF_8$ but for the pair
$(0,0)$.

\begin{lem} The diversity $\divsrc$ of a $\hsrc(n,k)$ is $(n-1)/2$.
\end{lem}
\begin{IEEEproof}
We have that $n=q^d-1$ for some suitable $d$.
The polynomial $p(x)$ is evaluated in $\alpha = \sum_{i=0}^{d-1}a_i w^i$, where $a_i \in \FF_q$ and $(a_0,...,a_{d-1})$ takes all the possible $q^d$ values, but for the whole zero one. Thus for every $\alpha$, we can create the pairs $(\alpha+\beta,\beta)$ where $\beta$ takes $q^d-2$ possible values, that is all values besides 0 and $\alpha$. This gives $q^d-2$ (which is equal to $n-1$) pairs, but since pairs $(\alpha+\beta,\beta)$ and  $(\beta,\alpha+\beta)$ are equivalent, we have $(n-1)/2$ distinct such pairs.
\end{IEEEproof}

An interesting property of SRC can be inferred from its diversity.
\begin{cor}
\label{cor:2isenough}
For a Homomorphic SRC, if at least $(n+1)/2$ fragments are available, then for any of the unavailable fragments, there exists some pair of available fragments which is adequate to reconstruct the unavailable fragment.
\end{cor}
\begin{IEEEproof}
Consider any arbitrary missing fragment $\alpha$. If up to $(n-1)/2$ fragments were available, in the worst case, these could belong to the $(n-1)/2$ exclusive pairs. However, if an additional fragment is available, it will be paired with one of these other fragments, and hence, there will be at least one available pair with which $\alpha$ can be reconstructed.
\end{IEEEproof}

\subsection{Overheads of recreating one specific missing fragment}

Recall that $x$ is defined as the number of fragments of an object that are available at a given time point. For any specific missing fragment, any one of the corresponding mutually exclusive pairs is adequate to recreate the said fragment. From Corollary \ref{cor:2isenough} we know that if $x \geq (n+1)/2$ then two downloads are enough. Otherwise, we need a probabilistic analysis. Both nodes of a specific pair are available with probability $(x/n)^2$. The probability that only two fragments are enough to recreate the missing fragment is $p_2 = 1-(1- (x/n)^2)^\divsrc$.
%\footnote{Such pairs may be probed either sequentially or in parallel. In all cases, the overheads of pinging are marginal with respect to the actual repair costs, and may also be amortized using the usual regular probing necessary in the system to detect missing fragments. For the sake of completeness: If pinging is done in parallel, the communication overhead is $2 \divsrc$ ping messages. The probability that $i$ pairs need to be pinged sequentially in order to encounter a pair with both nodes online is $(1-(x/n)^2)^{i-1}*(x/n)^2$. The expected number of ping messages will be $\sum_{i=1}^{\divsrc}(1-(x/n)^2)^{i-1}*(x/n)^2$.}

If two fragments are not enough to recreate a specific fragment, it may still be possible to reconstruct it with larger number of fragments. A loose upper bound can be estimated by considering that if 2 fragments are not adequate, $k$ fragments need to be downloaded to reconstruct a fragment,\footnote{Note than in fact, often
fewer than $k$ fragments will be adequate to reconstruct a specific fragment.} which happens with
  a probability $1-p_2 = (1- (x/n)^2)^\divsrc$.%\footnote{A somewhat tighter upper bound can be obtained by using $p_{obj}-p_2$.} %During the probing process, we have also determined which nodes are online, and which are not. So no further probings are required\footnote{Assuming that the time for completing the probes is relatively small w.r.to the level of dynamics in the system.}

%The probability that exactly $y>2$ fragments will be necessary to reconstruct the specific missing fragment for a given value of $x$ is $p_y = ???$. \fixme{If we can not determine this, then we can still find an upper-bound by stating that $k$ blocks are surely enough to recreate a missing block, since $k$-blocks are even enough to recreate the object.}

Thus the expected number $D_x$ of fragments that need to be downloaded to recreate one fragment, when $x$ out of the $n$ encoded fragments are available, can be determined as:
\begin{eqnarray*}
D_x=2 & \mbox{if } x\geq(n+1)/2\\
D_x<2p_2 + k(1-p_2)& \mbox{if } x < (n+1)/2.
\end{eqnarray*}

\subsection{Overhead of recreating all missing fragments}

Above, we studied the overheads to recreate one fragment. All the missing fragments may be repaired, either in parallel (distributed in different parts of the network) or in sequence. If all missing fragments are repaired in parallel, then the total overhead $D_{prl}$ of downloading necessary fragments is: $$D_{prl} = (n-x) D_x.$$

If they are recreated sequentially, then the overhead $D_{seq}$ of downloading necessary fragments is: $$D_{seq} = \sum_{i=x}^{n} D_i.$$

In order to directly compare the overheads of repair for different repair strategies - eager, or lazy parallelized and lazy sequential repairs using SRC, as well as lazy repair with traditional erasure codes, consider that lazy repairs are triggered when a threshold $x=x_{th}$ of available encoded fragments out of $n$ is reached. If eager repair were used for SRC encoded objects, download overhead of $$D_{egr}= 2(n-x_{th})$$ is incurred. Note that, when SRC is applied, the aggregate bandwidth usage for eager repair as well as both lazy repair strategies is the same, assuming that the threshold for lazy repair $x_{th} \geq (n+1)/2$.

In the setting of traditional erasure codes, let us assume that one node downloads enough ($k$) fragments to recreate the original object, and recreates one fragment to be stored locally, and also recreates the remaining $n-x_{th}-1$ fragments, and stores these at other nodes. This leads to a total network traffic: $$D_{EClazy}= k+n-x_{th}-1.$$ Eager strategy using traditional erasure codes will incur $k$ downloads for each repair, which is obviously worse than all the other scenarios, so we ignore it in our comparison.

Note that if less than half of the fragments are unavailable, as observed in Corollary \ref{cor:2isenough}, downloading two blocks is adequate to recreate any specific missing fragment. When too many blocks are already missing, applying a repair strategy analogous to traditional erasure codes, that of downloading $k$ blocks to recreate the whole object, and then recreate all the missing blocks is logical. That is to say, the benefit of reduced maintenance bandwidth usage for SRC (as also of other recent techniques like RGC) only makes sense under a regime when not too many blocks are unavailable. Let us define $x_c$ as the critical value, such that if the threshold for lazy repair in traditional erasure codes $x_{th}$ is less than this critical value, then, the aggregate fragment transfer traffic to recreate missing blocks will be less using the traditional technique (of downloading $k$ fragments to recreate whole object, and then replenish missing fragments) than by using SRC. Recall that for $x \geq (n+1)/2$, $D_{egr}=D_{prl}=D_{seq}$. One can determine $x_c$ as follows. We need $D_{egr} \leq D_{EClazy}$, implying that
\[
 2n-2x_{c} \leq n-1+k-x_{c}\Rightarrow x_c=n+1-k.
\]
Figure \ref{fig:trafficperlostblock} shows the average amount of network traffic to transfer data from live nodes per lost encoded fragment when the various lazy variants of repair are used, namely parallel ($\gamma_{prl}$) and sequential ($\gamma_{seq}$) repairs with SRC, and (by default, sequential) repair ($\gamma_{eclazy}$) when using EC. RGCs are also shown on this figure
- see $\gamma_{MSRGC}$, where $d$ is the number of live nodes contacted during repair.

The \emph{x-axis} represents the threshold $x_{th}$ for lazy repair, such that repairs are triggered only if the number of available blocks for an object is not more than $x_{th}$. Use of an eager approach with SRC incurs a constant overhead of two fragments per lost block. Note that there are other messaging overheads to disseminate necessary meta-information (e.g., which node stores which fragment), but we ignore these in the figure, considering that the objects being stored are large, and data transfer of object fragments dominates the network traffic. This assumption is reasonable, since for small-objects, it is well known that the meta-information storage overheads outweigh the benefits of using erasure codes, and hence erasure coding is impractical for small objects.

\begin{figure}
\begin{center}
\includegraphics[scale=0.5]{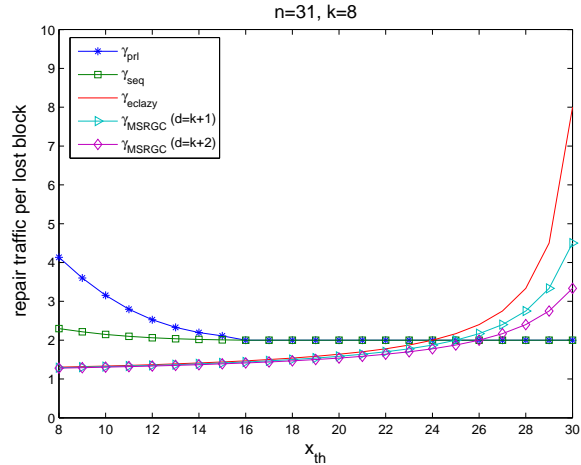}
\caption{Average traffic  normalized with $B/k$ per lost block for various choices of $x_{th}$.}
\label{fig:trafficperlostblock}
\end{center}\vspace{-6mm}
\end{figure}

There are several implications of the above observed behaviors. To start with, we note that an engineering solution like lazy repair which advocates waiting before repairs are triggered, amortizes the repair cost per lost fragment, and is effective in reducing total bandwidth consumption and outperforms SRC (in terms of total bandwidth consumption), provided the threshold of repair $x_{th}$ is chosen to be lower than $x_c$. This is in itself not surprising. However, for many typical choices of $(n,k)$ in deployed systems such as $(16,10)$ in Cleversafe \cite{cleversafe}, or $(517,100)$ in Wuala \cite{wuala}, a scheme like SRC is practical. In the former scenario, $x_c$ is too low, and waiting so long makes the system too vulnerable to any further failures (i.e., poor system health). In the later scenario, that is, waiting for hundred failures before triggering repairs seems both unnecessary, and also, trying to repair 100 lost fragments simultaneously will lead to huge bandwidth spikes.\footnote{A storage system's vulnerability to further failures, as well as spiky bandwidth usage are known problems of lazy repair strategies \cite{dattaSSS}.}

Using SRC allows for a flexible choice of either an eager or lazy (but with much higher threshold $x_{th}$) approaches to carry out repairs, where the repair cost per lost block stays constant for a wide range of values (up till $x_{th}\geq(n+1)/2$). Such a flexible choice makes it easier to also benefit from the primary advantage of lazy repair in peer-to-peer systems, namely, to avoid unnecessary repairs due to temporary churn, without the drawbacks of (i) having to choose a threshold which leads to system vulnerability or (ii) choose a much higher value of $n$ in order to deal with such vulnerability, and (iii) have spiky bandwidth usage.

\subsection{Maximal distance separability \& minimum storage point}\label{subsec:RGC}

To conclude the discussion on the cost of repair, this subsection gives some points of comparison between the now well known RGC and the newly introduced SRC.
The theory underlying regenerating codes exposes an interesting trade-off between the storage and repair bandwidth overheads for maximal distance separable codes - where the data is encoded and stored over $n$ nodes, and encoded data stored at any arbitrary $k$ of these storage nodes allows reconstruction of the whole object.

Suppose that each node has a storage capacity of $\alpha$, i.e., the size of the encoded data block stored at a node is of the size $\alpha$. When one data block needs to be regenerated, a new node contacts $d$ live nodes, and downloads $\beta$ amount of data from each of the contacted nodes (referred to as the bandwidth capacity of the connections between any node pair). By considering an information flow from the source to the data collector, a trade-off between the nodes' storage capacity and bandwidth is computed through a min-cut bound. This analysis determines two interesting constraints. Firstly, regeneration of a lost node is feasible only when at least $k$ live nodes are contacted, i.e. $d \geq k$. Secondly, it determines a trade-off curve between the storage overhead per node $\alpha$ and the bandwidth per regeneration $d \beta$. One extreme of this trade-off curve corresponds to the smallest feasible value of $\alpha$, which is $B/k$ (the other one being the smallest feasible value of $\beta$, called the minimal bandwidth repair point (MBR)). We note that this storage overhead corresponds to any optimal encoding scheme which aims to reconstruct the object using no more than $k$ encoded blocks. This point on the trade-off curve is called the \emph{minimum storage repair point} (MSR), determining the minimal bandwidth requirement for regeneration according the min-cut max-flow arguments of information flow as follows:
\[
(\alpha_{MSR},\beta_{MSR})=
\left(
\frac{B}{k},\frac{B}{k(d-k+1)}
\right).
\]
A similar minimum storage point, computed using the same type of arguments, is available for collaborative RGCs, and takes the
form
%\begin{equation}\label{eq:mscr}
\[
\alpha=\frac{B}{k},~\beta=\beta'=\frac{B}{k}\frac{1}{d-k+t}
%\end{equation}
%(\alpha_{MSR},\beta_{MSR},\beta'_{MSR})=,
\]
where $\beta'_{MSR}$ denotes the bandwidth used for cooperation and $t$ is the number of new nodes regenerating together and in cooperation with each other (thus, $t$ could be interpreted as the number of failures triggering lazy, collaborative repair).

Since encoded blocks in $\hsrc(n,k)$ codes are also of size $B/k$, a meaningful comparison is possible corresponding to the minimal storage (MSR) point. We note that firstly, HSRC achieves a $d << k$, which breaches the $d \geq k$ constraint of RGCs. Furthermore, we notice from Figure \ref{fig:trafficperlostblock} that HSRC can carry out regenerations with less bandwidth per repair than RGCs, for certain number of faults, and certain RGC parameter choices.

At first sight, this may seem counter-intuitive, given that the max-flow min-cut analysis establishes hard achievability constraints. But recall that these constraints were determined under the assumption of maximal distance separability of the resulting code. Thus to say, the significantly superior HSRC based repair performance is obtained by relaxing the MDS constraint, which, as discussed in Section \ref{sec:comapreMDS}, has marginal practical drawbacks or overheads.

We will like to add furthermore, that while the max-flow min-cut bound determines achievability constraint, it does not indicate what coding scheme may achieve the same. The currently existing RGC coding schemes in fact often do not support arbitrary values of $d$. Typical codes in literature are for $d = n-1$ (for MBR) or $d=k+1$ (for MSR), and often support only a single repair at a time, which means, for all practical purposes, the benefit of using HSRC is even stronger. This is because, even if there are regimes where RGCs do better (as may seem from Fig. \ref{fig:trafficperlostblock}), those can not in practice be reached with any currently known RGC codes.

The proposed HSRC code also does not have any hidden constraint on the underlying field size considered. This is in contrast to the implicit assumptions in RGC that both a suitable MDS erasure code and network code exist, whose existence typically relies on the ability of finding solutions to given systems of linear equations. Given a choice of parameters $(n,k)$, an MDS code does not exist for every field size. As for the system of linear equations, solutions tend to exist when the field size becomes big enough. Bounds on the field size for a network code to exist is a topic of ongoing study in the area of network coding.

Finally, it is worth pointing out the significance of the choice of $d$. A typical value of $k$ as used with Wuala is
about 100, this means that the number of nodes contacted for one repair is more than a hundred, whereas HSRC in contrast can repair one node by communicating with only two nodes. 

%*******************************************************************************************************%
%
% PARALLEL REPAIR
%
%*********************************************************************************************************%

\section{Other practical implications: A qualitative discussion}
So far we have demonstrated that by embracing the self-repairing properties, significant reduction in the aggregate bandwidth used for repairs is achieved. This overhead reduction is with respect to not only traditional erasure codes, but under certain regimes, also in comparison to other `optimal' storage centric codes such as regenerating codes. While repair bandwidth overhead reduction was the explicit motivation for designing self-repairing codes, the code properties have another natural and desirable consequence - in presence of multiple failures, SRC allows for fast and parallel repairs. We will elaborate this property with an example in next subsection.

Apart the lack of maximum distance separability (MDS) property, another possible critique of HSRC is that it is not a systematic code, that is, pieces of the object are not present uncoded. We have already argued that the original design goal of the self-repairing properties themselves are mutually exclusive with the MDS property, but based on quantitative arguments (see Section \ref{sec:comapreMDS}), we concluded that this has marginal impact on the resilience or storage overheads of the proposed code. We note that the systematic code property is not necessarily and completely exclusive of the cardinal self-repairing code properties. Indeed, a different construct of self-repairing code \cite{OD2-11} based on very different mathematical properties, that of projective geometry, has been shown to have systematic-like features. Later in this section, we provide some qualitative arguments on why the lack of systematic property may not have significant implications in terms of decoding, precisely because of the strong self-repairing properties; besides highlighting that some real life storage system deployment even intentionally avoid using systematic encoded blocks in order to enhance security.

\subsection{Fast \& parallel repairs with HSRC}

We observed in the previous section that while SRC is effective in significantly reducing bandwidth usage to carry out maintenance of lost redundancy in coding based distributed storage systems, depending on system parameter choices, an engineering solution like lazy repair while using traditional EC may (or not) outperform SRC in terms of total bandwidth usage, even though using lazy repair with EC entails several other practical disadvantages.

A final advantage of SRC which we further showcase next is the possibility to carry out repairs of different fragments independently and in parallel (and hence, quickly). If repair is not fast, it is possible that further faults occur during the repair operations, leading to both performance deterioration as well as, potentially, loss of stored objects.

Consider the following scenario for ease of exposition: Assume that each node in the storage network has an uplink/downlink capacity of 1 (coded) fragment per unit time. Further assume that the network has relatively (much) larger aggregate bandwidth. Such assumptions correspond reasonably with various networked storage system environments.
%In peer-to-peer systems, individual users have orders of magnitude less bandwidth w.r.to the total bandwidth in the network. In data centers, edge nodes typically have 1GE connections, while the aggregation and core layers have 10GE connections.

Consider that for the Example \ref{ex:complete}, originally $n$ was chosen to be $n_{max}$, that is to say, a $\hsrc(15,3)$ was used. Because of some reasons (e.g., lazy repair or correlated failures), let us say that seven encoded fragments, namely $p(1),\ldots,p(w^6)$ are unavailable while fragments $p(w^7)...p(w^{15})$ are available. Table \ref{tab:enumeratereconstruction} enumerates possible pairs to reconstruct each of the missing fragments.

\begin{table}
\begin{tabular}{|c|c|}
  \hline
  fragment & suitable pairs to reconstruct\\
  \hline
  $p(1)$   & $(p(w^7),p(w^9))$;$(p(w^{11}),p(w^{12}))$\\
  $p(w)$   & $(p(w^7),p(w^{14}))$;$(p(w^8),p(w^{10}))$\\
  $p(w^2)$ & $(p(w^7),p(w^{12}))$;$(p(w^9),p(w^{11}))$;$(p(w^{12}),p(w^{10}))$\\
  $p(w^3)$ & $(p(w^8),p(w^{13}))$;$(p(w^{10}),p(w^{12}))$\\
  $p(w^4)$ & $(p(w^9),p(w^{14}))$;$(p(w^{11}),p(w^{13}))$\\
  $p(w^5)$ & $(p(w^7),p(w^{13}))$;$(p(w^{12}),p(w^{14}))$\\
  $p(w^6)$ & $(p(w^7),p(w^{10}))$;$(p(w^8),p(w^{14}))$\\
  \hline
\end{tabular}
\vspace{1mm}
\caption{Scenario: Seven fragments $p(1),\ldots,p(w^6)$ are missing}
\label{tab:enumeratereconstruction}\vspace{-8mm}
\end{table}

A potential schedule to download the available blocks at different nodes to recreate the missing fragments is as follows: In first time slot, $p(w^{11})$, $p(w^{10})$, $p(w^{12})$, nothing, $p(w^{13})$, $p(w^{7})$ and $p(w^{8})$ are downloaded separately by seven nodes trying to recreate each of $p(1),\ldots,p(w^6)$ respectively. In second time slot $p(w^{12})$, $p(w^{8})$, $p(w^{7})$, $p(w^{10})$, $p(w^{11})$, $p(w^{13})$ and $p(w^{14})$ are downloaded. Note that, besides $p(w^3)$, all the other missing blocks can now already be recreated. In third time slot, $p(w^{12})$ can be downloaded to recreate it. Thus, in this example, six out of the seven missing blocks could be recreated within the time taken to download two fragments, while the last block could be recreated in the next time round, subject to the constraints that any node could download or upload only one block in unit time.

\begin{table*}
\center
\begin{tabular}{|c|c|c|c|c|c|c|c|}
  \hline
  node & $p(w^0)$ & $p(w^1)$ & $p(w^2)$ & $p(w^3)$ & $p(w^4)$ & $p(w^5)$ & $p(w^6)$ \\
  \hline
  Time 1 & $p(w^7)$ & $p(w^8)$ & $p(w^9)$ & $p(w^{13})$ & $p(w^{11})$ & $p(w^{12})$ & $p(w^{10})$ \\
  Time 2 & $p(w^9)$ & $p(w^{10})$ & $p(w^{11})$ & $p(w^8)$ & $p(w^{13})$ & $p(w^{14})$ & $p(w^7)$ \\
  \hline
\end{tabular}
\end{table*}

Even if a full copy of the object (hybrid strategy \cite{Liskov}) were to be maintained in the system, with which to replenish the seven missing blocks, it would have taken seven time units. While, if no full copy was maintained, using traditional erasure codes would have taken at least nine time units.

This example demonstrates that SRC allows for fast reconstruction of missing blocks. Orchestration of such distributed reconstruction to fully utilize this potential in itself poses interesting algorithmic and systems research
challenges which we intend to pursue as part of future work.
%Such work can additionally take into account network topology to determine schemes to optimally place the encoded fragments.

\subsection{On HSRC not being a systematic code}
\label{subsec:syst}

One can immediately read partial contents of the stored object from systematic encoded blocks. This has both
advantages - for object retrieval, and disadvantages - in terms of security.

The fact that HSRC is not a systematic code (every encoded fragment contains information about every piece of
data) makes the object retrieval more costly: this is basically decoding. However, unlike in a classical communication
scenario where decoding has to be done with whatever corrupted data is available, the situation is different here:
thanks to the repair property, it is possible to have a privileged set of encoded fragments to be used for decoding, and
if some are missing during object retrieval, they can be repaired first. The set of encoded fragments to decode can be
chosen for being closer to systematic blocks than random blocks, or precomputed computations can be made available to ease the decoding.

Having a systematic code can cause security threats.
If the data is to be stored securely over untrusted storage nodes - as may be the case in peer-to-peer systems, but also in cloud/data-center environments which may be partially compromised by either malicious insiders or hackers, then one would need to apply cryptographic techniques on the original object, and then store the encrypted object. However, if non-systematic blocks are used, then individual encoded blocks do not reveal any information. Instead, one would need access to enough ($k$) encoded blocks before being able to read any (and in fact, the whole) content. Thus, use of non-systematic encoded blocks provides some level of protection, similar in spirit to threshold cryptography - in that, even if a small subset of the storage nodes are compromised, it does not reveal any content. While the level of security is not at par with encryption/decryption by the data owner using a secret key, this nevertheless provides an intermediate degree of protection, but without the additional cost of encryption/decryption, instead amortizing on the encoding/decoding overheads. In fact, Cleversafe employs this principle for protection according to their corporate website, providing an example of distributed storage scenario, where use of systematic encoded blocks are deemed undesirable, and in fact, non-systematic blocks are preferred.

This security argument, coupled with our arguments above on how reasonably efficient decoding is possible, thus mitigating the adversarial impact of the lack of systematic property of HSRC, indicates that HSRC is practical for many storage centric application scenarios.

\section{Conclusion}
We propose a new family of codes, called self-repairing codes, which are designed by taking into account specifically the characteristics of distributed networked storage systems. Self-repairing codes achieve excellent properties in terms of maintenance of lost redundancy in the storage system, most importantly: (i) low-bandwidth consumption for repairs (with flexible/somewhat independent choice of whether an eager or lazy repair strategy is employed), (ii) parallel and independent (thus very fast) replenishment of lost redundancy. When compared to erasure codes, the self-repairing property is achieved by marginally compromising on static resilience for same storage overhead, or conversely, utilizing marginally more storage space to achieve equivalent static resilience. This paper provides the theoretical foundations for SRCs, and shows its potential benefits for distributed storage. There are several algorithmic and systems research challenges in harnessing SRCs in distributed storage systems, e.g., design of efficient decoding algorithms, or placement of encoded fragments to leverage on network topology to carry out parallel repairs, which are part of our ongoing and future work.

%SIMULATIONS:
%* range: one object = up to 64 megabytes
%* example from wala: (517,100)=(n,k) code
%* R-S: example from wiki, typically k=223 over F8 (=8 bit symbols)

%************************************************************************%
%
% ACK
%
%***********************************************************************%

\section*{Acknowledgment}

F. Oggier's research  for this work has been supported by the Singapore National Research Foundation under Research Grant NRF-CRP2-2007-03. A. Datta's research for this work has been supported by AcRF Tier-1 grant number RG 29/09.

%********************************************************************%
%
% BIBLIO
%
%********************************************************************%

\end{document}